\documentclass[10pt,conference]{IEEEtran}

\usepackage{graphicx}
\usepackage{epstopdf}
\usepackage{amsfonts,amssymb,amsbsy,amsthm}
\usepackage{amsmath}
\usepackage{bbm}
\usepackage[percent]{overpic}
\usepackage{url}
\usepackage{color}

\newtheorem{theorem}{Theorem}

\newtheorem{lemma}{Lemma}
\newtheorem{proposition}{Proposition}
\newtheorem{definition}{Definition}
\newtheorem{remark}{Remark}

\DeclareMathOperator{\mat}{Mat}

\newcommand{\R}{\mathbb{R}}
\newcommand{\E}{\mathbb{E}}
\newcommand{\Z}{\mathbb{Z}}

\newcommand{\info}{\mathrm{I}}

\begin{document}

\title{Information Bounds and Flatness Factor Approximation for Fading Wiretap {MIMO} Channels}

\author{Amaro Barreal$^\ast$, Alex Karrila$^\ast$, David Karpuk and Camilla Hollanti \\
Department of Mathematics and Systems Analysis, School of Science, Aalto University, Finland. \\
email: firstname.lastname@aalto.fi}
\maketitle

{\let\thefootnote\relax\footnotetext{$^\ast$ Equal contribution.}}

\begin{abstract}

In this article, the design of secure lattice coset codes for general wireless channels with fading and Gaussian noise is studied. Recalling the eavesdropper's probability and information bounds, a variant of the latter is given from which it is explicitly seen that both quantities are upper bounded by (increasing functions of) the expected flatness factor of the faded lattice related to the eavesdropper.
    
By making use of a recently developed approximation of the theta series of a lattice, it is further shown how the average flatness factor can be approximated numerically. In particular, based on the numerical computations, the average flatness factor not only bounds but also orders correctly the performance of different lattices.
\end{abstract}

\section{Introduction}
\label{sec:intro}
In the wireless wiretap scheme two legitimate communication parties, Alice and Bob, exchange information in the presence of an \emph{eavesdropper}, Eve. In this setting, the communication parties rely on physical layer security rather than cryptographic protocols. Hence, Eve is assumed to have no computational limitations and know the cryptographic key, if any, but to have a worse signal quality than Bob. 

The objective of code design in a wiretap channel is to maximize the data rate and Bob's correct decoding probability while minimizing Eve's information. It was shown in the seminal paper of Wyner \cite{wyner} that the legitimate parties can design codes with asymptotically non-zero rate, zero error probability and zero information leakage. Today, this setup is particularly interesting in wireless channels that are open in nature but vulnerable to distortions

As a practical construction of a wiretap code, \cite{wyner-ozarow} introduced the general technique of \emph{coset coding}, where random bits are added to the message to confuse the eavesdropper. In the specific case of a wireless channel, where lattice codes are suitable, the code lattice $\Lambda_b$ is endowed with a sublattice $\Lambda_e \subset \Lambda_b$ which carries the random bits \cite{oggier}.

\subsection{Related Work and Contributions}
\label{subsec:related}

The security of lattice coset codes can be quantized either by Eve's correct decision probability, or alternatively by the mutual information of the message and Eve's received signal. For the \emph{additive white Gaussian noise} (AWGN) channel, upper bounds are known for both approaches \cite{oggier, ling} and, more importantly, both are increasing functions of the \emph{flatness factor} of the lattice $\Lambda_e$. Sequences of lattice coset codes achieving security and reliability are also constructed in \cite{ling}.

For fading channel models, probability and information bounds were derived in \cite{belfiore, belfiore_mimo}, and \cite{mirghasemi, luzzi}, respectively. Codes achieving security and reliability in the \emph{multiple-input multiple-output} (MIMO) channel were given in \cite{luzzi}. In this paper, we recall the strategy of \cite{belfiore, belfiore_mimo} and give a variant of the information bounds. With this streamlined introductory computation, we obtain bounds which are increasing functions of the expected flatness factor of the faded lattice related to the eavesdropper. The agreement of the probability and information approaches is hence explicit, and a natural and explicit generalization of the AWGN case. The computations hold in any channel model with linear fading and Gaussian noise. We remark that, to the best of our knowledge, the steps towards practical code designs in low dimensions, such as \cite{gnilke}, are based on the probability bounds.

Motivated by this, we show how to use an approximation of the theta series of a lattice, recently derived in \cite{barreal}, to efficiently compute the average flatness factor of a given lattice. Hence, we do not need to rely on further approximations for the information/probability bounds, \emph{e.g.}, the common approach using the \emph{inverse norm sum} in SISO channels \cite{belfiore}. We exemplify this in the Rayleigh fast fading channel, and see an agreement between our numerical computations and the geometric design heuristics and simulations results in \cite{gnilke}. In particular, this agreement supports the expectation that the flatness factor not only \emph{bounds} but also \emph{orders correctly} the performance of different lattices, hence serving as a design criterion for practical low-dimensional constructions, as first suggested in \cite{belfiore, belfiore_mimo}. Comparing to \cite{gnilke}, where average flatness factor heuristics are tested with simulations, our results suggest that it is indeed approachable for design heuristics as well as numerical computations.

\section{Lattices, Theta Series and the Flatness Factor}
\label{sec:lattices}
A \emph{lattice} $\Lambda \subset \R^n$ is a discrete subgroup of $\R^n$ with the property that there exist $s \le n$ linearly independent vectors $\left(\mathbf{b}_1,\ldots,\mathbf{b}_s\right)$ of $\R^n$ such that 
    \begin{align*}
        \Lambda = \bigoplus\limits_{i=1}^{s}{\mathbf{b}_i \Z}.
    \end{align*} 
We say that $\left(\mathbf{b}_1,\ldots,\mathbf{b}_s\right)$ is a $\Z$-basis of $\Lambda$, and call $s \le n$ the \emph{rank}, and $n$ the \emph{dimension} of $\Lambda$. The lattice is \emph{full} if $s = n$.

A lattice $\Lambda' \subset \R^n$ such that $\Lambda' \subset \Lambda$ is called a \emph{sublattice} of $\Lambda$, and $\Lambda$ is referred to as a \emph{superlattice} for $\Lambda'$.  

For a convenient presentation, we define a \emph{generator matrix} $M_{\Lambda} := \left[\mathbf{b}_1\  \cdots\ \mathbf{b}_s\right] \in \mat(n\times s,\R)$, and equivalently write
\begin{align*}
    \Lambda = \left\{\left.\lambda = M_\Lambda \mathbf{z} \right| \mathbf{z} \in \Z^s \right\}.
\end{align*}

\begin{definition}
    Let $\Lambda \subset \R^n$ be a full lattice, and let $\lambda_{\min} := \min_{\lambda \in \Lambda\backslash\left\{\mathbf{0} \right\}}||\lambda||^2$ be its \emph{minimal norm}. The lattice $\Lambda$ is called \emph{well-rounded} if the set $\left\{\left. \lambda \in \Lambda\ \right| ||\lambda||^2 = \lambda_{\min} \right\}$ contains $n$ linearly independent vectors. 
\end{definition}

The \emph{volume} of $\Lambda$ is defined to be $\nu_{\Lambda} = \left|\det(M_{\Lambda})\right|$, and is independent of the choice of basis. If $\Lambda$ is not full, then $\nu_{\Lambda} = \det(M_{\Lambda}^t M_{\Lambda})^{1/2}$. The dual lattice $\Lambda^\ast$ of a full lattice $\Lambda$ is the lattice generated by $M_{\Lambda^\ast} := (M_\Lambda^{-1})^t = (M_\Lambda^{t})^{-1}$. We can easily compute the volume of a sublattice\footnote{The index $|\Lambda/\Lambda'|$ is finite provided that $\dim(\Lambda) = \dim(\Lambda')$.} $\Lambda' \subset \Lambda$ and of the dual lattice $\Lambda^\ast$, as
\begin{align*}
    \nu_{\Lambda'} = \nu_{\Lambda}\left|\Lambda/\Lambda'\right|;\quad
    \nu_{\Lambda^\ast} = 1/\nu_{\Lambda},
\end{align*}
where $|\Lambda/\Lambda'|$ is the group \emph{index} of $\Lambda'$ in $\Lambda$. Further, the \emph{Voronoi cell} associated with a lattice point $\lambda \in \Lambda$ is the set 
\begin{align*}
    \mathcal{V}_{\Lambda}(\lambda) := \left\{\left. \mathbf{x} \in \R^n \right| ||\mathbf{x}-\lambda||^2 \le ||\mathbf{x} - \lambda'||^2, \lambda' \in \Lambda \backslash \left\{\lambda\right\} \right\},
    \end{align*}
and $\mathcal{V}(\Lambda) := \mathcal{V}_{\Lambda}(\mathbf{0})$ denotes the \emph{basic Voronoi cell} of $\Lambda$. 
    
\begin{definition}
\label{def:theta}
Let $\Lambda$ be a lattice. The \emph{theta series} (or theta function) of $\Lambda$ is the generating function
\begin{align*}
    \Theta_{\Lambda}(q) := \sum\limits_{\lambda\in \Lambda}{q^{\left|\left| \lambda\right|\right|^2}}.
\end{align*}
\end{definition}
The theta series $\Theta_\Lambda(q)$ converges absolutely if $0 \le q < 1$. In this article, we will need to compute the theta series of lattices which have been affected by random fading, \emph{i.e.}, with \emph{random} generator matrices. Unfortunately, there is no known way of computing the theta series of a random lattice in closed form. The following result is thus crucial for our purposes.

\begin{proposition}\cite[Prop.~1]{barreal}
\label{prop:theta_approx}
    Let $\Lambda \subset \R^n$ be a full lattice with fundamental volume $\nu$ and minimal norm $\lambda_{\min}$. The theta series $\Theta_{\Lambda}\left(e^{-\pi\tau}\right)$, as a function of $\tau$, can be expressed as 
    \begin{align*}
        \Theta_{\Lambda}\left(e^{-\pi\tau}\right) = 1 + \frac{(\pi\lambda_{\min})^{\frac{n}{2}+1}\tau}{\Gamma\left(\frac{n}{2}+1\right)\nu}\int\limits_{1}^{\infty}{t^{\frac{n}{2}}e^{-\pi\tau\lambda_{\min} t} dt} + \Xi_n,
    \end{align*}
    where $\Xi_n = \Xi_n(\tau,\Lambda,L)$ denotes the error term. 
\end{proposition}
We refer to \cite{barreal} for a more detailed version of this result. In Figure~\ref{fig:theta_approx} we illustrate the accuracy of this approximation for various famous lattices for which the theta series is known in closed form and thus can be computed explicitly.
\begin{figure}[!h]
    \includegraphics[width=0.43\textwidth]{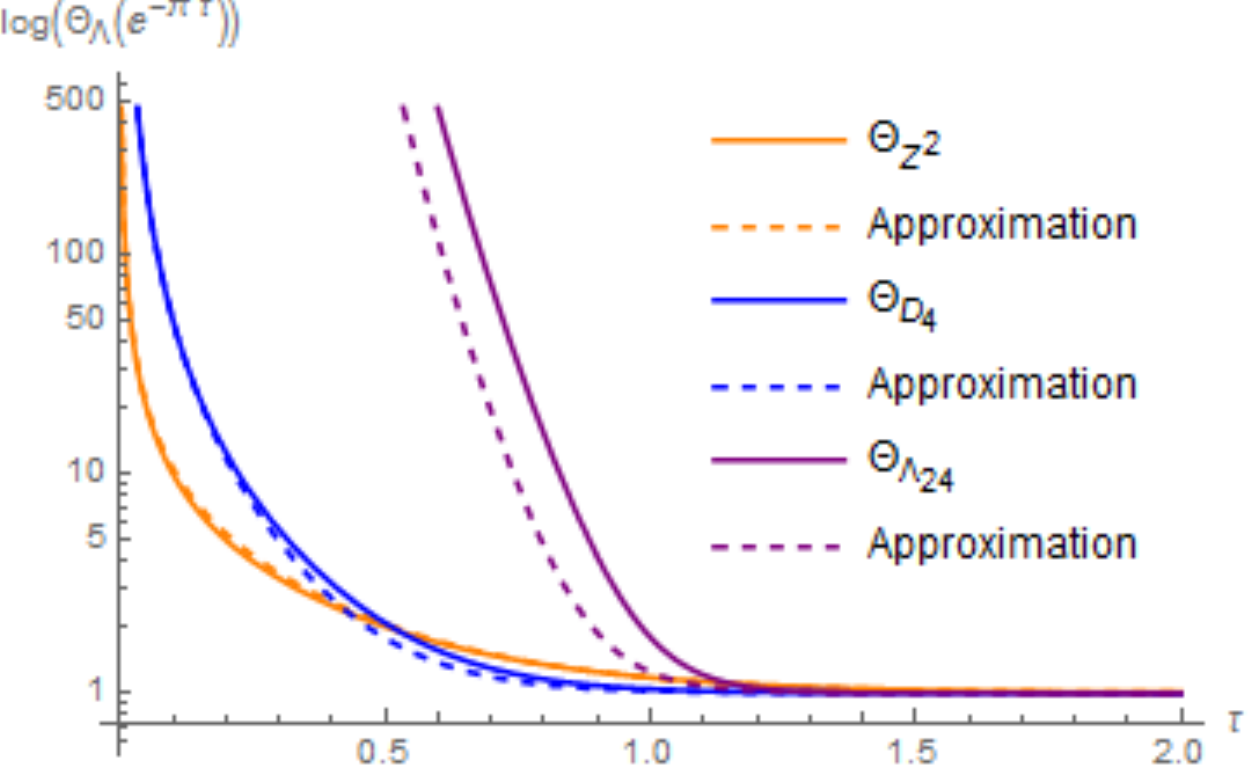}
    \caption{Approximation of the theta series of the lattices $\mathbb{Z}^2$, $D_4$ and the Leech lattice $\Lambda_{24}$, in dimensions $n = 2, 4$ and $24$.}
    \label{fig:theta_approx}
\end{figure}

\subsection{Lattice Sums and Flatness Factor}
\label{subsec:ff}
Let us denote the \emph{probability density function} (PDF) of the $n$-dimensional spherical Gaussian as
\begin{align*}
    g_n(\mathbf{x};\sigma) = \frac{1}{(\sqrt{2\pi}\sigma)^n}\exp\left(-\frac{||\mathbf{x}||^2}{2\sigma^2}\right),
\end{align*}    
and its sums over a (possibly shifted) lattice $\Lambda$ as 
\begin{align*}
    g_n(\Lambda+\mathbf{x};\sigma) := \sum\limits_{\lambda \in \Lambda}g_n(\lambda+\mathbf{x};\sigma).
\end{align*}
We then have the identity
\begin{align*}
    g_n(\Lambda;\sigma) = \frac{1}{(\sqrt{2\pi}\sigma)^n}\Theta_{\Lambda}\left(e^{-\frac{1}{2\sigma^2}}\right),
\end{align*} 
from which we see that $g_n(\Lambda;\sigma)$ only differs by constants from the standard theta series. It is easy to see that $g_n(\Lambda+\mathbf{x};\sigma)$ is $\Lambda$-periodic, and, for full lattices, it defines a PDF on the Voronoi cell $\mathcal{V}(\Lambda)$, called the \emph{lattice Gaussian PDF}. 

The following definition is crucial for all subsequent results.
\begin{definition}
Let $\Lambda \subset \R^n$ be a lattice generated by $M_\Lambda \in \mat(n,\R)$, with fundamental volume $\nu_\Lambda$. The \emph{flatness factor} $\varepsilon_{\Lambda}(\sigma)$ of $\Lambda$ is defined as 
    \begin{align*}
        \varepsilon_{\Lambda}(\sigma) = \varepsilon_{M_\Lambda}(\sigma) := \max \limits_{\mathbf{x} \in \R^n}\left|\frac{g_n(\Lambda+ M_\Lambda \mathbf{x};\sigma)}{1/\nu_\Lambda}-1\right|.
    \end{align*}
\end{definition}

The flatness factor was introduced in \cite{ling} as a wiretap information tool and measures the deviation of the lattice Gaussian PDF 
from the uniform distribution on $\mathcal{V}(\Lambda)$. 
As the maximum of $g_n(\Lambda+\mathbf{x};\sigma)$ is attained for $\mathbf{x} \in \Lambda$ \cite{ling}, the flatness factor of a full lattice can be expressed in terms of theta series, 
\begin{align}
\label{eqn:ff_theta}
    \varepsilon_{\Lambda}(\sigma) &= \nu_{\Lambda} g_n(\Lambda;\sigma)  - 1 \\
\label{eqn:ff_theta dual}
    &=  \Theta_{\Lambda^\ast}\left(e^{-2\pi\sigma^2}\right)-1.
\end{align}
Note that the additive group structure of a lattice is actually crucial in this work: these important formulas are based on the Poisson summation formula.

\section{System Model and Coset Codes}
\label{sec:system}

We consider a wireless fading channel with noise. Perfect channel state information is assumed at both receivers (CSIR), Bob and Eve; the transmitter is only assumed to know the channel statistics. As we are only interested in the eavesdropper's performance, we henceforth only consider the channel between Alice and Eve, and consequently forgo subscripts in the related quantities. Throughout this paper, random variables are denoted by capital letters and their realizations with lower-case letters.
Denote Alice's transmitted vector by $\mathbf{x} \in \R^n$, so that the channel equation is given by 
\begin{align*}
    \mathbf{y} = \mathbf{h} \mathbf{x}+\mathbf{n},
\end{align*}
where $\mathbf{h} \in \mat(m \times n,\R)$ is the realization of the fading, and the noise vector $\mathbf{N} \in \R^m$ is composed of i.i.d. components $N_i \sim \mathcal{N}(0,\sigma^2)$. We assume that $\mathbf{H}$ has full rank almost surely, but need not be a square matrix.

\begin{remark} Typically a complex fading channel model is considered, together with complex lattice codes. Such codes, however, can be be reduced to the real case with double the dimension. We consider here a real channel model since the class of real lattices is wider than that of complex lattices, most importantly including complex $\Z[e^{2 \pi i / 3}]$ -lattices, and since the average theta functions, to which both probability and information bounds reduce, are already computed explicitly \cite{belfiore, belfiore_mimo}.
\end{remark}

To confuse Eve, Alice uses \emph{lattice coset coding}. An original information vector thus corresponds not only to a lattice point, but to an entire coset, and the transmitted vector is then chosen randomly within the coset according to a distribution discussed below. More specifically, Alice is equipped with two nested lattices\footnote{The notation $\Lambda_b$, $\Lambda_e$ is chosen to indicate that the message intended for Bob is taken from $\Lambda_b$, while $\Lambda_e$ is the lattice that is chosen to confuse Eve.} $\Lambda_e \subset \Lambda_b \subset \R^n$, as well as an injective map from her message space $\mathcal{M}$ of cardinality $|\mathcal{M}| = |\Lambda_b/\Lambda_e|$, into the set of unique coset representatives of $\Lambda_b/\Lambda_e$, 
\begin{align*}
    \mathcal{E}: \mathcal{M} \to \Lambda_b \cap \mathcal{V}(\Lambda_e), \quad
    M \mapsto \lambda_M.
\end{align*}
Alice then chooses a representative of the coset class corresponding to $M$ at random, \emph{i.e.}, picks $\lambda_e \in \Lambda_e$ and transmits $\mathbf{x} = \lambda_M + \lambda_e \in (\lambda_M+\Lambda_e) \in \Lambda_b/\Lambda_e$. The transmitted vector now represents the original information bits, and in addition, contains random bits which are encoded by $\Lambda_e$. 

If $\mathbf{x} = M_\Lambda\mathbf{z}$ for $\mathbf{z} \in \Z^n$, then $\mathbf{h} \mathbf{x} = \mathbf{h} M_\Lambda\mathbf{z}$, and we can think of a lattice code under fading with CSIR as a Gaussian-channel lattice code where the code lattice realizes a random lattice with generator matrix $\mathbf{h} M_\Lambda$. We will henceforth denote the faded lattices $\Lambda_b$ and $\Lambda_e$ by $\Lambda_{b,\mathbf{h}}$ and $\Lambda_{e,\mathbf{h}}$, respectively. 

We denote the mutual information of two random variables $X, Y$ by $\info\left[X;Y\right]$, and for double conditions given a third random variable $Z$, we write $\info\left[X;(Y,Z)\right]$. The metric of interest is the information that Eve is able to extract from her observations, \emph{i.e.}, $\info\left[M;(\mathbf{Y,H})\right]$. The fading $\mathbf{H}$, noise $\mathbf{N}$ and Alice's message $M$ are assumed to be mutually independent.
The settings considered in this article are:
\begin{enumerate}
    \item Alice chooses coset class representatives uniformly at random, and the channel is a \emph{$\bmod\ \Lambda_s$} channel.
    
    \item Alice uses Gaussian coset coding.
\end{enumerate}
In the simplest setup, Alice chooses uniform representatives from a finite transmission region, and Eve's information is affected by the boundaries of that region. Naively speaking, if Eve has knowledge about the transmission region, she can often guess the transmitted vector correctly if the received vector lies outside of the transmission region. Our choice of setups can be roughly seen as removing such boundary effects by a modulo operation, and smoothing the boundary of the transmission region, respectively. The AWGN information bounds have been derived for these setups in \cite{ling}, and the information-theoretic results for the fading channels in \cite{mirghasemi, luzzi} are for the Gaussian coset coding setup.

\section{Information and Probability Bounds}
\label{sec:inf_bounds}

Let $\Lambda_b$ and $\Lambda_e$ be of dimension $n$, and assume that Eve simply decodes the received signal to the closest lattice point in $\Lambda_{e, \mathbf{h}}$. Then, probability of Eve correctly decoding the message is upper bounded by \cite{belfiore, belfiore_mimo}
\begin{align}
\label{eqn:p_eve}
& \mathbb{P}\left[\text{Decode correctly in } \Lambda_e:\Lambda_b \text{ coset code} \right] \nonumber \\
= &\mathbb{E}_{\mathbf{H}}\left[\mathbb{P}\left[\text{Decode correctly in } \Lambda_{e,\mathbf{h}}:\Lambda_{b,\mathbf{h}} \text{ AWGN coset code}\right]\right] \nonumber \\
\le &\mathbb{E}_{\mathbf{H}}\left[ \nu_{\Lambda_{b, \mathbf{h}}} g_n(\Lambda_{e, \mathbf{h}} ; \sigma) \right] \nonumber \\
= &\left[\Lambda_b : \Lambda_e\right]^{-1}\left(\mathbb{E}_{\mathbf{H}}\left[\varepsilon_{\Lambda_{e, \mathbf{h}}}(\sigma) \right] + 1\right).
\end{align}
The first step is Fubini's theorem and the independence of $\mathbf{N}$, $\mathbf{H}$ and $M$, while the second step follows from the probability bound in the AWGN setup \cite{oggier}. The bound is given for Rayleigh fading \emph{single-input single-output} (SISO) and MIMO channels in \cite{belfiore} and \cite{belfiore_mimo}, respectively. We recall that the flatness factor is defined through \eqref{eqn:ff_theta} for non-full lattices as well. 

We will give a variant of the information bounds \cite{ mirghasemi, luzzi} from which it is explicit that the bounds agree (up to constants) with this upper bound on the probability that Eve correctly decodes the message. Hence, to minimize this probability, it will be necessary to minimize the upper bound on the mutual information, and vice versa.

\begin{lemma}\cite[Lemma~2]{ling}
\label{lem:var_distance}
    Let $\mathbf{Y}$ be an $\R^n$-valued random variable with PDF $\rho_{\mathbf{Y}}$, and consider a message space $\mathcal{M}$ such that $\vert \mathcal{M} \vert \ge 4$. Suppose that there exists some PDF $\tilde{\rho}$ in $\R^n$ such that for all $m \in \mathcal{M}$, the variational distance 
\begin{align*}
    V(\rho_{\left\{\mathbf{Y}\vert M=m\right\}},\tilde{\rho}) := \int_{\mathbf{y} \in \R^n} \vert \rho_{\left\{\mathbf{Y} \vert M=m\right\}}(\mathbf{y}) - \tilde{\rho}(\mathbf{y}) \vert d^n \mathbf{y}
\end{align*}
is upper bounded\footnote{This assumption is implicit in \cite{ling} but necessary, seen by taking $\delta \to \infty$. This yields a small difference between our and earlier information bounds.} by $\delta \le e^{-1}/2$. Then,
\begin{align*}
    I\left[M; \mathbf{Y}\right] \le 2 \delta \log(\vert \mathcal{M} \vert) - 2 \delta \log(2 \delta ).
\end{align*}    
\end{lemma}

\subsection{The $\bmod$ $\Lambda_s$ Wiretap Channel}
\label{subsec:mod_lambda_channel}
We first consider a strategy where the boundary of the shaping region set to be the boundary of $\mathcal{V}(\Lambda_s)$, where $\Lambda_s$ is a third lattice $\Lambda_s \subset \Lambda_e \subset \Lambda_b \subset \R^n$, called the \emph{shaping lattice}. The random shifting vector $\lambda_e$ is chosen uniformly at random from the set of representatives of $\Lambda_e/\Lambda_s$ in $\mathcal{V}(\Lambda_s)$. Then, slightly artificially, we assume that Eve only receives knowledge of the equivalence class $\mathbf{y}/\Lambda_{s,\mathbf{h}}$. 

\begin{theorem}
\label{thm:inf_bound_mod_lambda}
In the $\bmod$ $\Lambda_s$ channel setup, let the message $M$ have any distribution on the message space $\mathcal{M}$ of cardinality $\vert \mathcal{M} \vert \ge 4$, and assume that $\E_{\mathbf{H}}\left[\varepsilon_{\Lambda_{e, \mathbf{h}}}(\sigma)\right] \le e^{-1}/2$. Then,
\begin{align*}
\info\left[M; (\mathbf{Y}/\Lambda_{s, \mathbf{h}} , \mathbf{H})\right] &\le 2 (e + 1) \E_{ \mathbf{H} }  \left[ \varepsilon_{\Lambda_{e, \mathbf{h}}} (\sigma)  \right] \log(\vert \mathcal{M} \vert) \\
&-2 \E_{ \mathbf{H} }  \left[ \varepsilon_{\Lambda_{e, \mathbf{h}}} (\sigma)  \right]  \log(2 \E_{ \mathbf{H} }  \left[ \varepsilon_{\Lambda_{e, \mathbf{h}}} (\sigma)\right]). 
\end{align*}
\end{theorem}

\begin{proof}
By the independence assumptions, we have the identity
\begin{align*}
    \info\left[M; (\mathbf{Y}/\Lambda_{s, \mathbf{h}} , \mathbf{H})\right] = \E_\mathbf{H} \left[ \info\left[M ; (\mathbf{Y}/\Lambda_{s, \mathbf{h}}  \vert \mathbf{H}=\mathbf{h})\right]\right].
\end{align*}
We divide $\mathbf{Y}$ into components $\mathbf{Y}_{\perp}$ and $\mathbf{Y}_{\parallel}$, perpendicular and parallel to the nested lattices $\Lambda_{*, \mathbf{h}}$. By the independence assumptions, given $\mathbf{H}=\mathbf{h}$, $\mathbf{Y}_{\perp}$ is independent of the transmitted lattice point and hence of the message $M$. Thus, $\info\left[M ; (\mathbf{Y}/\Lambda_{s, \mathbf{h}}  \vert \mathbf{H}=\mathbf{h})\right] = \info\left[M ; (\mathbf{Y}_{\parallel}/\Lambda_{s, \mathbf{h}}  \vert \mathbf{H}=\mathbf{h})\right]$.
For a fixed channel realization $\mathbf{h}$, the channel is just a Gaussian $\bmod$ $\Lambda_{s, \mathbf{h}}$ channel. Then, given the message $m$, the projection modulo Eve's lattice, $\mathbf{Y}_{\parallel}/\Lambda_{e, \mathbf{h}}$, has the lattice $\Lambda_{e, \mathbf{h}}$ Gaussian distribution, of which the PDF of $\mathbf{Y}_{\parallel}/\Lambda_{s, \mathbf{h}}$ is simply a scaling. Denoting $\varepsilon = \varepsilon_{\Lambda_{e, \mathbf{h}}} (\sigma)$, Lemma~\ref{lem:var_distance} and the trivial upper bound $\log (\vert \mathcal{M} \vert)$ on the mutual information yield
\begin{align*}
    \info & \left[M; (\mathbf{Y}/\Lambda_{s, \mathbf{h}} , \mathbf{H})\right]  \\
 &\le \E_\mathbf{H}\left[\mathbbm{1}_{\left\{\varepsilon \le e^{-1} / 2\right\}} (2 \varepsilon \log(\vert \mathcal{M} \vert) 
 - 2 \varepsilon \log(2 \varepsilon))\right] \\
 &\quad + \E_{\mathbf{H}}\left[\mathbbm{1}_{\left\{\varepsilon > e^{-1}/2\right\}} \log(\vert \mathcal{M} \vert)\right] \\
&= \mathbb{P}_{\mathbf{H}} \left[ \varepsilon \le e^{-1} / 2 \right] \E_{\left\{\mathbf{H} \vert \varepsilon \le e^{-1} / 2\right\}}\left[ 2 \varepsilon \log(\vert \mathcal{M} \vert) - 2 \varepsilon \log(2 \varepsilon ) \right]\\ 
&\quad + \E_{\mathbf{H}} \left[ \mathbbm{1}_{\left\{\varepsilon > e^{-1} / 2\right\}}\right] \log(\vert \mathcal{M} \vert).
\end{align*}

We upper bound each of the summands separately. For the first one, we apply Jensen's inequality to the convex function $x \log x$ and bound $\mathbb{P}_{\mathbf{H}} \left[ \varepsilon \le e^{-1} / 2 \right] \le 1$, which yields
\begin{align*}
    &\mathbb{P}_{\mathbf{H}} \left[ \varepsilon \le e^{-1} / 2 \right] \E_{\left\{\mathbf{H} \vert \varepsilon \le e^{-1} / 2\right\}}\left[ 2 \varepsilon \log(\vert \mathcal{M} \vert) - 2 \varepsilon \log(2 \varepsilon ) \right] \\
   &\le 2 \E_{\left\{ \mathbf{H} \vert \varepsilon \le e^{-1} / 2\right\}}  \left[ \varepsilon \right]\left(\log(\vert \mathcal{M} \vert) - \log(2 \E_{ \left\{\mathbf{H} \vert \varepsilon \le e^{-1} / 2\right\}}\left[\varepsilon\right])\right)  \\
   & \le 2 \E_{ \mathbf{H}}\left[\varepsilon\right]\left(\log(\vert \mathcal{M} \vert) - \log(2 \E_{\mathbf{H}}\left[\varepsilon\right])\right).
\end{align*}
The second inequality holds since $0 \le \E_{\left\{\mathbf{H} \vert \varepsilon \le e^{-1} / 2\right\}}\left[\varepsilon\right] \le \E_{\mathbf{H}}\left[\varepsilon\right] \le e^{-1}/2$, and $x \log x$ is decreasing in this interval.

For the second term, we can use Markov's inequality to get
\begin{align*}
    \E_{\mathbf{H}}\left[\mathbbm{1}_{\left\{\varepsilon > e^{-1} / 2\right\}}\right] \le \frac{\E_{\mathbf{H}}\left[\varepsilon\right]}{e^{-1} / 2} =2e \E_{\mathbf{H}}\left[\varepsilon\right].
\end{align*}
The result follows.
\end{proof}

\subsection{Discrete Gaussian Coset Coding}
\label{subsec:discrete_gaussian_channel}

While an insightful scenario, the restrictions imposed in the $\bmod$ $\Lambda_s$ channel setting are not necessarily realistic. To be more general, we consider a second approach, where the boundary is smoothed instead of removed. Here, the vectors $\lambda_e$ corresponding to the random bits are chosen so that the message $\mathbf{X} = \lambda_M + \lambda_e$ follows the Gaussian distribution centered on the shifted lattice $\Lambda_e + \lambda_M$, that is, for all $\mathbf{x} \in \Lambda_e + \lambda_M$,
\begin{align*}
    P\left[\mathbf{X} = \mathbf{x}\right] = \frac{g_n(\mathbf{x};\sigma_s)}{g_n(\Lambda_e+\lambda_M;\sigma_s)} =: D_{\Lambda_e,\lambda_M}(\mathbf{x};\sigma_s).
\end{align*}
The parameter $\sigma_s^2$ is called the \emph{shaping variance}. 

The following result, which can be regarded as a special case of \cite[Lemma~1]{luzzi} in our notation, is needed for the generalization of the information bound to this second setting.

\begin{lemma}
\label{lem:discrete_continuous_gaussian}
Fix $\mathbf{h} \in \mat(m \times n,\R)$ and let $\mathbf{X}$ have the centered discrete Gaussian distribution $D_{\Lambda_e,\lambda_M} (\mathbf{x}; \sigma_s) $, where $\Lambda_e$ is full. Let $\mathbf{N} \sim \mathcal{N}(\mathbf{0}, \sigma^2 \mathbf{I}_{m})$ be a spherical (continuous) Gaussian vector independent of $\mathbf{X}$. Assume furthermore that $\varepsilon_{\sqrt{\sigma^2/\sigma_s^2 I_{n} +  \mathbf{h}^t\mathbf{h}}\Lambda_e}(\sigma) \le 1/2$, where $\sqrt{\cdot}$ denotes matrix square root. Then, the PDF $\rho (\mathbf{y})$ of $\mathbf{Y} = \mathbf{h} \mathbf{X} + \mathbf{N}$ and the PDF $\tilde{\rho}(\mathbf{y})$ of $\mathcal{N}(\mathbf{0}, (\sigma^2 I_{m} + \sigma_s^2 \mathbf{h}\mathbf{h}^t))$,
\begin{equation*}
\tilde{\rho} (\mathbf{y}) = \frac{\exp\left(-\frac{1}{2}  \mathbf{y}^t(\sigma^2 I_{m} + \sigma_s^2 \mathbf{hh}^t)^{-1} \mathbf{y}\right) }{(\sqrt{2 \pi})^m\sqrt{\det  (\sigma^2 I_{m} + \sigma_s^2 \mathbf{hh}^t)}},
\end{equation*}
have variational distance at most $4 \varepsilon_{\sqrt{\sigma^2/\sigma_s^2 I_{n} +  \mathbf{h}^t \mathbf{h}} \Lambda_e}(\sigma)$.
\end{lemma}

\begin{theorem}
Consider the Rayleigh fading channel with discrete Gaussian coset coding. Let the message $M$ have any distribution on the message space $\mathcal{M}$ of cardinality $\vert \mathcal{M} \vert \ge 4$. Assume that 
    $E := \E_{\mathbf{H}}\left[\varepsilon_{\sqrt{ \sigma^2 / \sigma_s^2 I_{n} +  \mathbf{h}^t \mathbf{h} } \Lambda_e}(\sigma)\right] \le \frac{1}{8e}.$
Then,
\begin{align*}
    \info\left[M; (\mathbf{Y} , \mathbf{H})\right] \le 8 (1+ e) E  \log(\vert \mathcal{M} \vert) - 8 E  \log\left(8 E \right).
\end{align*}
\end{theorem}

\begin{proof}
The proof closely follows the steps of that of Theorem~\ref{thm:inf_bound_mod_lambda}. We start by writing
\begin{align*}
    \info\left[M; (\mathbf{Y} , \mathbf{H})\right] = \E_\mathbf{H} \left[ \info\left[M ; (\mathbf{Y} \vert \mathbf{H}=\mathbf{h})\right] \right].
\end{align*}
For a fixed channel realization $\mathbf{h}$, by Lemma~\ref{lem:discrete_continuous_gaussian} the distribution of the received vector $\mathbf{Y}$ is close to a fixed Gaussian distribution $\tilde{\rho}$ for all messages $M$, with variational distance
\begin{align*}
    V\left(\rho_{\left\{\mathbf{Y}\vert M=m\right\}},\tilde{\rho}\right) \le 4  \varepsilon(\sigma),
\end{align*}
for $\varepsilon \le 1/2$. For the values $\mathbf{h}$ such that $\varepsilon \le \frac{1}{8e}$ we get an information bound using Lemma~\ref{lem:var_distance}. Otherwise, we have the trivial upper bound $\log(|\mathcal{M}|)$, yielding 
\begin{align}
\nonumber
    \info\left[M; (\mathbf{Y} , \mathbf{H})\right] &\le \E_\mathbf{H}\left[\mathbbm{1}_{\left\{\varepsilon \le \frac{1}{8e}\right\}}\left[8 \varepsilon \left( \log(\vert \mathcal{M} \vert) -  \log\left(8 \varepsilon \right)\right)\right] \right] \\
    \label{eq: two-term bound}
    &\quad +\E_{\mathbf{H}}\left[\mathbbm{1}_{\left\{\varepsilon > \frac{1}{8e}\right\}} \log(\vert \mathcal{M} \vert)\right].
\end{align}
The rest of the proof is identical to Theorem~\ref{thm:inf_bound_mod_lambda}, using Jensen's inequality for the first, and Markov's inequality for the second summand to obtain termwise bounds which we can substitute back to \eqref{eq: two-term bound} to conclude the proof.
\end{proof}

\subsection{Observations}
\label{subsec:observations}
The information bounds tend to zero with the respective average flatness factor. By the dual formula \eqref{eqn:ff_theta dual} for the flatness factor, $\varepsilon_{\Lambda}(\sigma)$ decreases monotonously to zero as $\sigma \to \infty$ for any $\Lambda$. As it only depends on ratios of the different parameters $\sigma, \sigma_s, \sigma_h$, it is easy to deduce that the respective average flatness factors also decrease monotonously to zero at poor signal quality. Hence, the bounds prove information-theoretic security for poor eavesdropper's channel quality.

The information bound of Theorem~2 involves the quantity 
\begin{align*}
    E :=\E_{\mathbf{H}}\left[\varepsilon_{\sqrt{ \sigma^2 / \sigma_s^2 I_{n} +  \mathbf{h}^t \mathbf{h} } \Lambda_e}(\sigma)\right].
\end{align*} 
In reasonable scenarios, we have $\sigma^2/\sigma_s^2 \ll 1$, and it is easily deduced that $E \underset{\sigma_s \to \infty} {\to} \E_{\mathbf{H}}\left[\varepsilon_{\sqrt{  \mathbf{h}^t \mathbf{h} } \Lambda_e}(\sigma)\right] = \mathbb{E}_{\mathbf{H}}\left[\varepsilon_{\Lambda_{e,\mathbf{h}}}(\sigma)\right].$ 
Hence, independent of the considered setup, the goal is to design a lattice $\Lambda_e$ so that $\mathbb{E}_{\mathbf{H}}\left[\varepsilon_{\Lambda_{e,\mathbf{h}}}(\sigma)\right]$ is minimized. We remark that this quantity only depends on $\Lambda_e$ and $\sigma_h/ \sigma$, not for example on $\Lambda_b$. Using \eqref{eqn:ff_theta}, we can also compute
\begin{align*}
    \E_{\mathbf{H}}\left[\varepsilon_{\Lambda_{e, \mathbf{h}}} (\sigma)\right] = |\Lambda_b/\Lambda_e| \E_{\mathbf{H}}\left[\nu_{\Lambda_{b, \mathbf{h}}} g_n(\Lambda_{e, \mathbf{h}} ; \sigma)\right] - 1,
\end{align*}
which coincides up to constants with the probability bound \eqref{eqn:p_eve}. This is an important agreement of the probability and information approaches.

We also point out that the agreeing probability bound \eqref{eqn:p_eve} has been analyzed with error terms in \cite{belfiore} and is asymptotic at poor signal quality. We thus expect the average flatness factor to predict and correctly order the performance of different lattices at the interesting low-SNR regime.

\section{Simulation Results}
\label{sec:sims}

In the following, we make use of Proposition~\ref{prop:theta_approx} to approximate the average flatness factor of different lattice coset codes numerically. The aim is to compare the predictions of our theta approximation to lattice design heuristics and channel simulations given \emph{e.g.}, in \cite{gnilke}. We focus on the SISO Rayleigh fading channel to allow for a comparison with the results obtained in \cite{gnilke}, and hence restrict the choice of channel matrices $\mathbf{h}$ to diagonal matrices. 

Even if Eve's information only depends on $\Lambda_e$, to ensure a meaningful comparison between codes we fix a superlattice $\Lambda_b$ (\textit{i.e.}, fix Bob's vector decoding error probability) and the index\footnote{As we are fixing a common superlattice fixing the index is equivalent to comparing sublattices of the same volume.} $|\Lambda_b/\Lambda_e|$ to match the information rate. For convenience, we define the variable $\tau = 1/(2\pi\sigma^2)$, and consider $\varepsilon_{\Lambda}$ as a function of $\tau$. In the limit of large and small values of $\tau$, our quantity of interest of two different lattices of equal index converge to the same value. Hence, we will center our attention to a range of $\tau$ where a difference is visible, which in our plots corresponds to a very large range of values for $\sigma^2$.

\begin{remark}
\label{rmk:diag}
   It has been recently shown in \cite{gnilke} that in order to minimize the expression \eqref{eqn:p_eve} for the SISO channel, the property of the considered lattice being well-rounded is favorable. As will be visible from the following simulation results, this criterion is also key for minimizing the average flatness factor. When comparing \emph{orthogonal} lattices of fixed volume, this reduces to taking a (close to) square lattice, \emph{i.e.}, $d\mathbb{Z}^n$ is expected to perform best among lattices with diagonal generator matrix and fixed volume $\nu = d^n$.

\end{remark}

We start by comparing two 2-dimensional lattices, $\Lambda_1$ and $\Lambda_2$ in Figure~\ref{fig:n2}, with generator matrices 
\begin{align*}
    M_{\Lambda_1} = \begin{bmatrix}
    4 & 0 \\ 0 & 4
    \end{bmatrix}, 
    M_{\Lambda_2} = \begin{bmatrix}
    1 & 0 \\ 0 & 16
    \end{bmatrix},
\end{align*}
which we interpret as index-16 sublattices of $\Lambda = \Z^2$. Some characteristics of these lattices are summarized in Table~\ref{tab:n2}. 

\begin{table}[!h]
\centering
    \begin{tabular}{|c||c|c|c|c|}
    \hline
        $\mathbf{n = 2}$ & $\lambda_{\min}$ & $\#\lbrace \lambda \in \Lambda \mid ||\lambda||^2 = \lambda_{\min}\rbrace$ & WR & Index \\
        \hline \hline
        $\Lambda_1$ & 16 & 4 & Yes & 16 \\
        \hline
        $\Lambda_2$ & 1 & 2 & No & 16 \\
        \hline
    \end{tabular}
    \caption{Characteristics of the lattices in dimension $n = 2$.}
    \label{tab:n2}
\end{table}

\begin{figure}[h!]
\centering
    \includegraphics[width=0.43\textwidth]{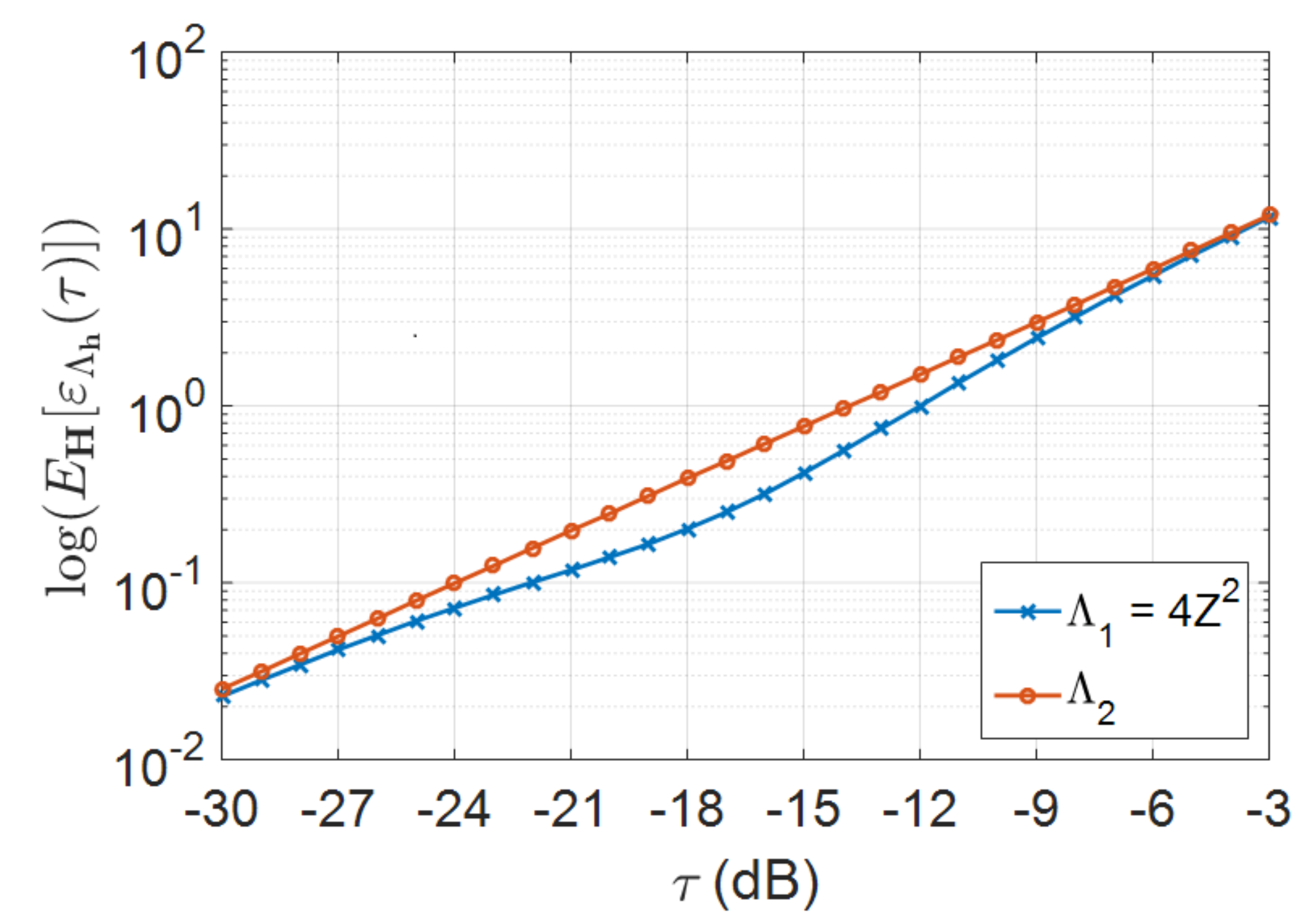}
    \caption{Average flatness factor of $\Lambda_1$ and $\Lambda_2$.}
    \label{fig:n2}
\end{figure}

The choice of $M_{\Lambda_2}$ is deliberately bad, hence from the observation in Remark~\ref{rmk:diag} it is not surprising that $\Lambda_1$ exhibits a lower average flatness factor than $\Lambda_2$. The difference can be as large as 3 dB, which is remarkable without further optimization. This example allows already to hint that both the property of being well-rounded as well as having a shortest vector which is as long as possible is advantageous for minimizing the average flatness factor. This statement will become clearer from the subsequent, more interesting examples. 

\begin{table}[!h]
\centering
    \begin{tabular}{|c||c|c|c|c|}
    \hline
        $\mathbf{n = 4}$ & $\lambda_{\min}$ & $\#\lbrace \mathbf{x} \in \Lambda \mid ||\mathbf{x}||^2 = \lambda_{\min}\rbrace$ & WR & Index \\
        \hline \hline
        $\Pi_1$ & 8 & 24 & Yes & 32 \\
        \hline
        $\Pi_2$ & 1 & 2 & No & 32 \\
        \hline \hline
        $\Omega_1$ & 4 & 4 & No & 256 \\
        \hline
        $\Omega_2$ & 16 & 8 & Yes & 256 \\
        \hline
        $\Omega_3$ & 20 & 12 & Yes & 256 \\
        \hline \hline
        $\Gamma_1$ & 4 & 2 & No & 302 \\
        \hline
        $\Gamma_2$ & 22 & 12 & Yes &  302 \\
        \hline
    \end{tabular}
    \caption{Characteristics of the lattices in dimension $n = 4$.}
    \label{tab:n4}
\end{table}

We consider $\Pi_1 = 2D_4$, a scaled version of the checkerboard lattice, as an index-32 sublattice of $\Z^4$, and compare its average flatness factor to another index-32 sublattice $\Pi_2$ of $\Z^4$. The generator matrices are given by
\begin{equation*}
    \resizebox{0.8\hsize}{!}{$M_{\Pi_1} = 2\cdot \begin{bmatrix} 
        -1 & 1 & 0 & 0 \\ -1 & -1 & 1 & 0 \\ 0 & 0 & -1 & 1 \\ 0 & 0 & 0 & -1
        \end{bmatrix}, 
    M_{\Pi_2} = \begin{bmatrix}
        2 & 0 & -2 & -2 \\ 3 & -5 & -2 & 3 \\ -4 & 8 & 4 & -4 \\ -6 & 5 & 0 & -7
    \end{bmatrix}.$}
\end{equation*}

In a recent article \cite{gnilke}, the authors analyze the performance of three different index-256 sublattices of $\Z^4$ with respect to Eve's decoding error probability. We denote the lattices by $\Omega_1$, $\Omega_2 = 4\Z^4$, and $\Omega_3$, with generator matrices given by
\begin{equation*}
    \resizebox{0.95\hsize}{!}{$M_{\Omega_1} = \begin{bmatrix} 
        16 & 0 & 0 & 0 \\ 0 & 4 & 0 & 0 \\ 0 & 0 & 2 & 0 \\ 0 & 0 & 0 & 2
    \end{bmatrix}, M_{\Omega_2} = \begin{bmatrix}
    4 & 0 & 0 & 0 \\ 0 & 4 & 0 & 0 \\ 0 & 0 & 4 & 0 \\ 0 & 0 & 0 & 4
    \end{bmatrix}, M_{\Omega_3} = \begin{bmatrix}
    -2 & -3 & 4 & -1 \\ 0 & -1 & 0 & 3 \\ 0 & -3 & -2 & -3 \\ -4 & -1 & 0 & -1
    \end{bmatrix}$.}
\end{equation*}
In Figure~\ref{fig:n4} we compare all five presented lattices with respect to their average flatness factor. Note that as remarked above, only a comparison between lattices of the same index is meaningful.
\begin{figure}[h!]
    \centering
    \begin{overpic}[width=0.43\textwidth]{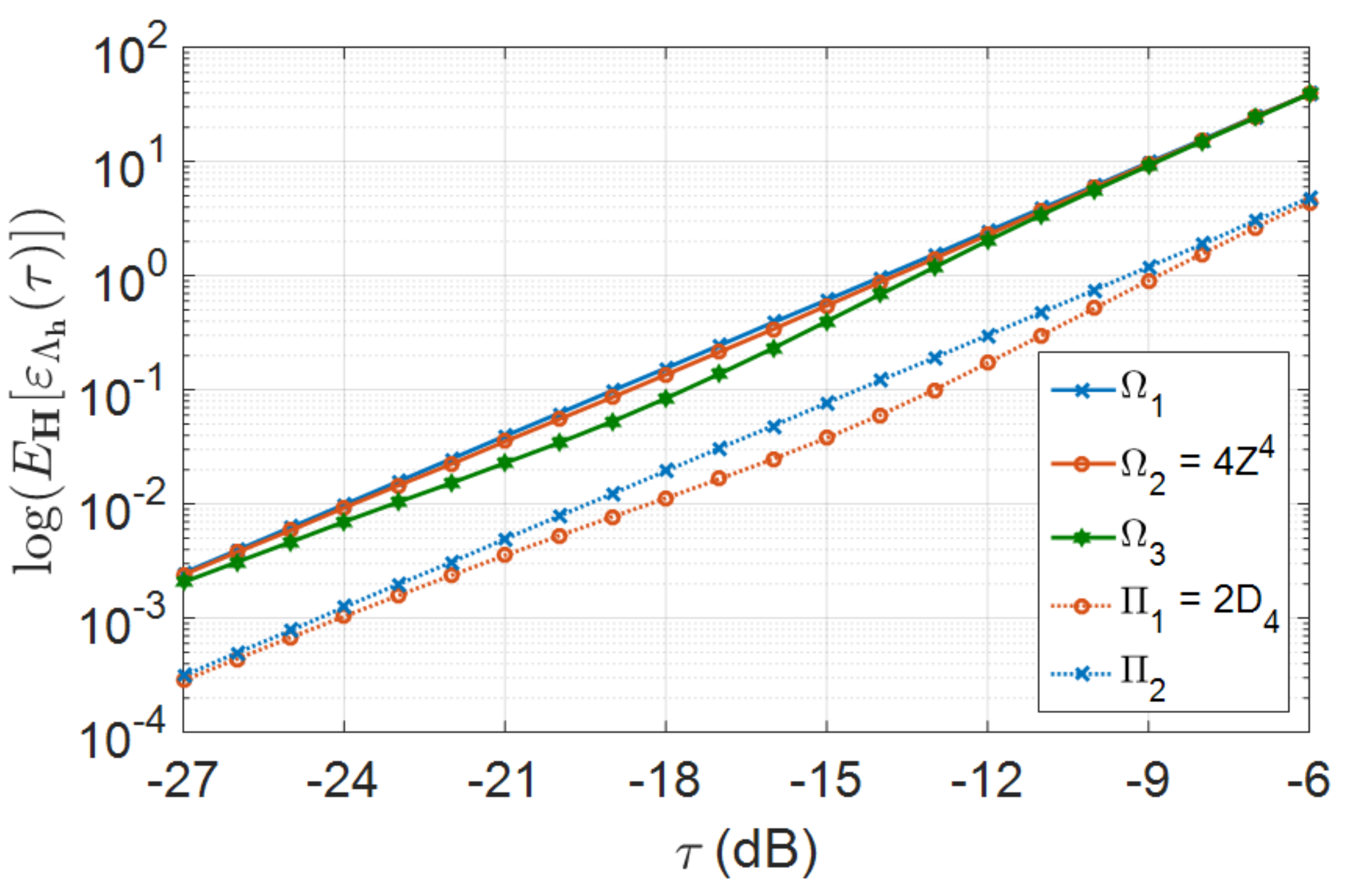}
        \put(15,41){\includegraphics[width=0.15\textwidth]{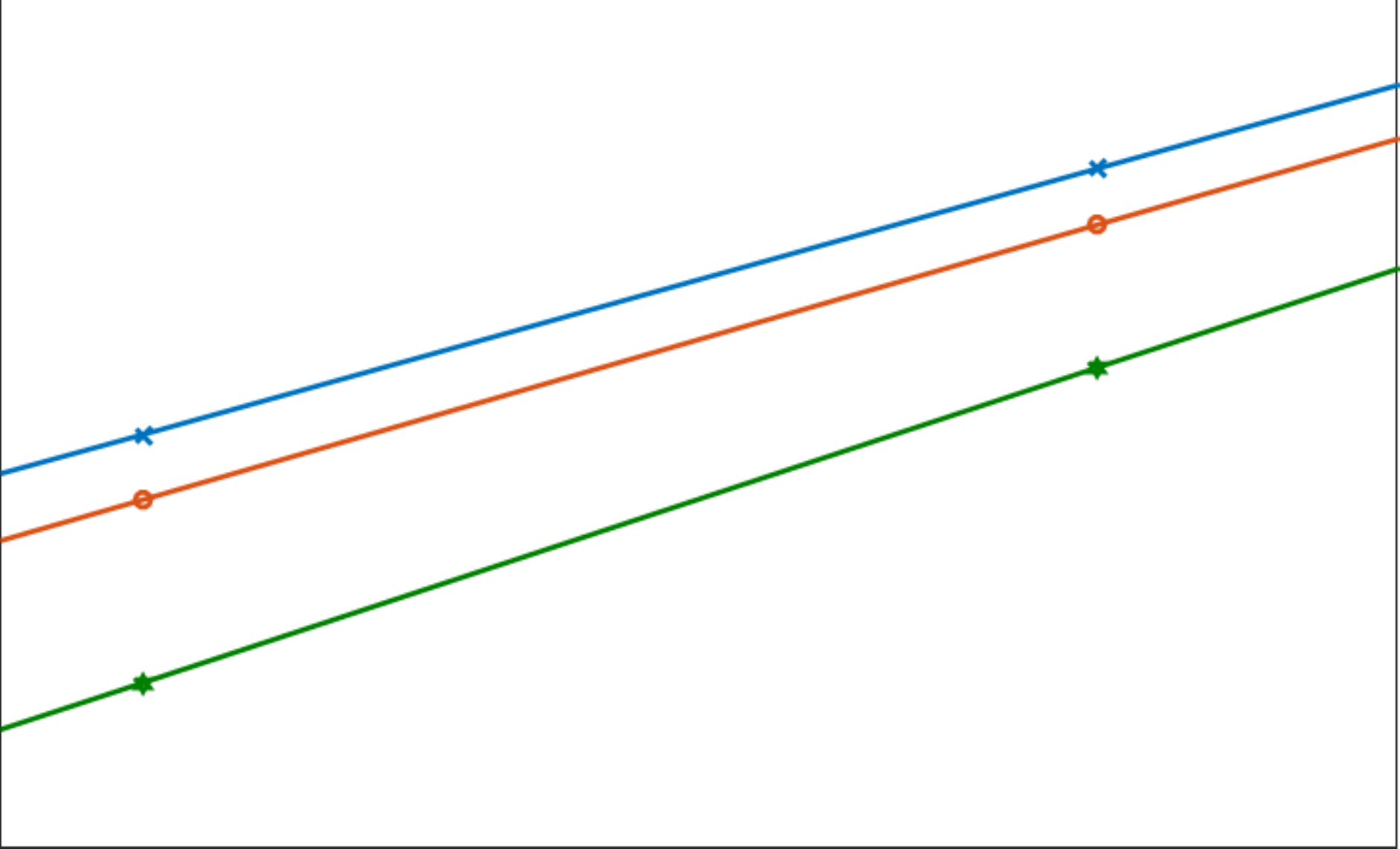}}
    \end{overpic}
    \caption{Average flatness factor of $\Pi_1$, $\Pi_2$, and of $\Omega_1$, $\Omega_2$, $\Omega_3$.}
    \label{fig:n4}
\end{figure}
The comparison between $\Pi_1$ and $\Pi_2$ again agrees with what should be expected. Note that it is known that $D_4$ attains the maximum possible length of the shortest vector among lattices in dimension $4$, and is moreover well-rounded. 

The order of the lattices $\Omega_i$, $1 \le i \le 3$, agrees with the results obtained in \cite{gnilke}. We note that $\Omega_1$ and $\Omega_2$ respect the statement from Remark~\ref{rmk:diag}, but their difference is almost negligible. The well-rounded lattice $\Omega_3$, however, is up to 1-1.5 dB ahead of $\Omega_2$. 

In addition, we compare two index-302 sublattices of $\Z^4$, of respective generator matrices
\begin{equation*}
    \resizebox{0.8\hsize}{!}{$M_{\Gamma_1} = \begin{bmatrix}
        -1 & 1 & 2 & 2 \\ -1 & 0 & 2 & -5 \\ 1 & -2 & 5 & -1 \\ -1 & -5 & 1 & 2
    \end{bmatrix}, 
    M_{\Gamma_2} = \begin{bmatrix}
        1 & 1 & 3 & -2 \\ -4 & 1 & 0 & -4 \\ -1 & -2 & 3 & 1 \\ 2 & -4 & 2 & -1 
    \end{bmatrix}.$}
\end{equation*}
\begin{figure}[h!]
\centering
\includegraphics[width=0.43\textwidth]{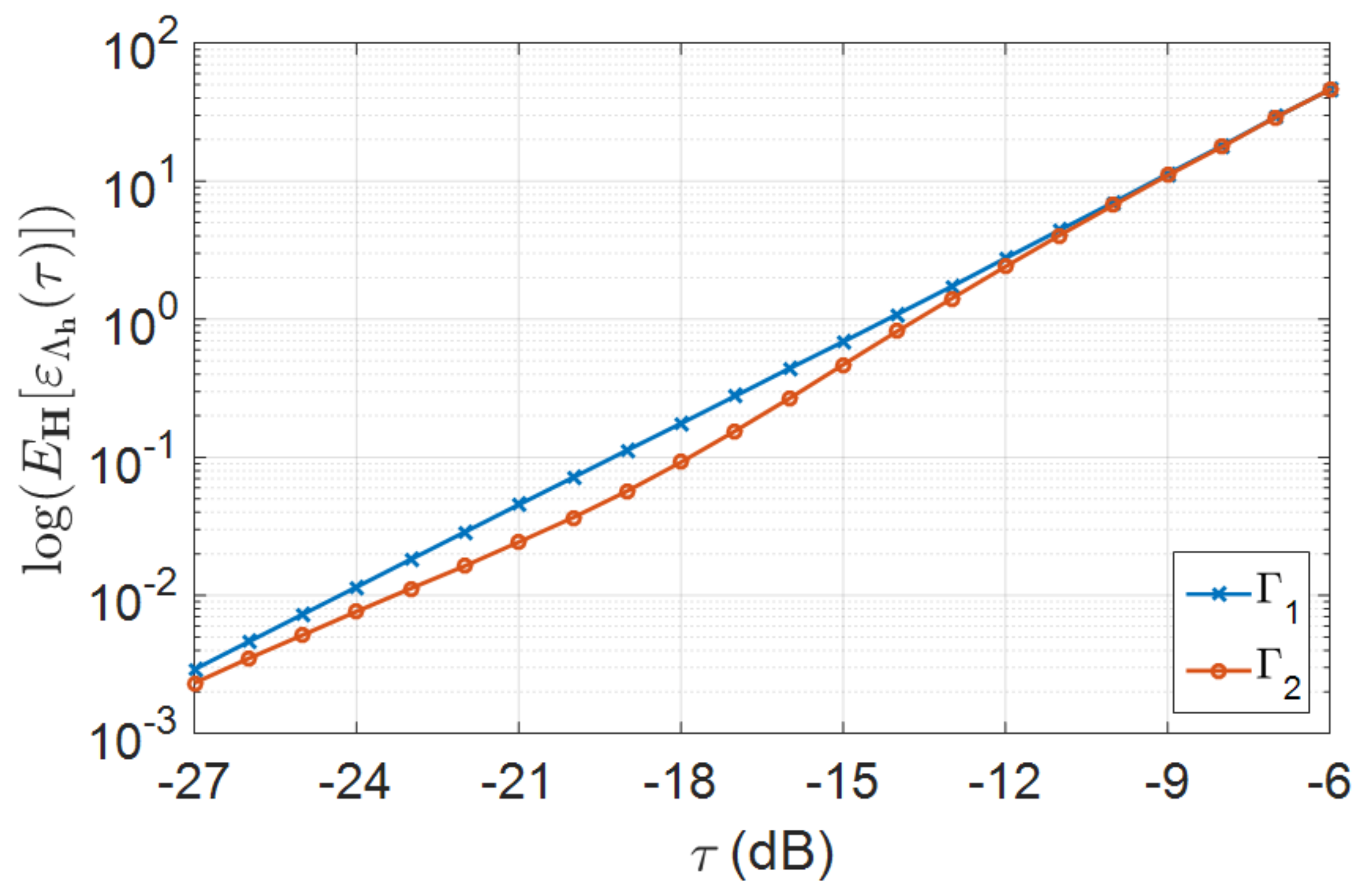}
    \caption{Average flatness factor of $\Gamma_1$ and $\Gamma_2$.}
\end{figure}

As in the previous examples, the well-rounded lattice $\Gamma_2$ exhibits a smaller average flatness factor than $\Gamma_1$, and the difference reaches up to 2 dB. This is in agreement with the simulation results obtained in \cite{gnilke}.

\section{Conclusions}
\label{sec:conclusions}
We studied the design of secure lattice coset codes in a general wireless model with any fading and Gaussian noise, and focused on the $\bmod\ \Lambda_s$ channel as well as on Gaussian coset coding. Recalling the eavesdropper's probability bounds \cite{belfiore, belfiore_mimo} and deriving a variant of the information bounds \cite{mirghasemi, luzzi} we saw that both are an increasing function of the average flatness factor, \emph{i.e.}, the theta series of the faded eavesdropper’s lattice $\Lambda_{e, \mathbf{h}}$ averaged over channel realizations $\mathbf{h}$.

We then computed numerically the average flatness factors of different faded lattices with help of an approximation of the theta series of a lattice derived in \cite{barreal}. By making use of this result, computing the average flatness factor becomes computationally inexpensive, which allows us to avoid further additional approximations based \emph{e.g.}, on the inverse norm sum. As already suggested \cite{gnilke}, our findings show that the property of being well-rounded, is necessary in order to minimize the average flatness factor and, hence, the upper bound on the eavesdropper's information. All our numerical computations agreed with the channel simulations in \cite{gnilke}. 

Interesting further topics of investigation include the more detailed design of optimal secure lattice codes and, for SISO channels, the construction of sequences of codes where the information tends to zero. In both of these cases, well-rounded lattices studied in \cite{gnilke} will likely play a central role.

\section*{Acknowledgment}
This work is supported by the Academy of Finland under Grants \#268364, \#276031, \#282938 and \#283262, as well as a grant from the Finnish Foundation for Technology Promotion.

We gratefully acknowledge the authors of \cite{luzzi} for bringing recent publications to our attention.

\end{document}